\documentclass[11pt]{article}

\usepackage[final]{acl}

\usepackage{times}
\usepackage{latexsym}

\usepackage[T1]{fontenc}

\usepackage[utf8]{inputenc}

\usepackage{microtype}

\usepackage{inconsolata}

\usepackage{graphicx}

\usepackage{hyperref}
\usepackage{hyperxmp}
\usepackage{listings}
\usepackage{amsmath}
\usepackage{amsfonts}
\usepackage{diagbox}
\usepackage{algorithmic}
\usepackage[ruled,vlined,linesnumbered,resetcount]{algorithm2e}
\usepackage{tabularx}
\usepackage{subcaption}
\usepackage{natbib}
\usepackage{multirow}
\usepackage{tikz}
\usepackage{pgfplots}
\usepackage{graphicx}
\usetikzlibrary{patterns}
\usepackage{enumitem}

\usepackage{balance}
\usepackage{booktabs}
\usepackage{amsthm}
\usepackage{comment}

\newtheorem{definition}{Definition}
\newcommand{\modelname}{RCTEA}

\title{\modelname{}: Richness-guided Co-training for Temporal Entity Alignment}

\author{
  Jiayun Li$^1$\quad Wen Hua$^{2}$\thanks{Corresponding author.}\quad Shiqi Fan$^2$\quad Fengmei Jin$^2$\quad Haiyang Jiang$^1$\quad Xue Li$^1$ \\[4pt]
  $^1$The University of Queensland, Brisbane, Queensland, Australia \\
  $^2$The Hong Kong Polytechnic University, Hong Kong SAR, China \\[4pt]
  \texttt{\{uqjli48, haiyang.jiang\}@uq.edu.au}\\
  \texttt{\{wency.hua, shiqi.fan, fengmei.jin\}@polyu.edu.hk}\\
  \texttt{xueli@eecs.uq.edu.au}
}

\usepackage{bibentry}

\begin{document}

\maketitle

\begin{abstract}
Temporal Entity Alignment (TEA), which aims to identify equivalent entities across Temporal Knowledge Graphs (TKGs), is crucial for integrating knowledge facts from multiple sources. However, existing TEA models often fail to capture the orthogonal yet complementary effect between structural and temporal features, and typically overlook the importance of information richness—a key factor for effective message passing in the neural feature encoders. To address these limitations, we propose a \modelname{} framework that jointly models both structural and temporal aspects of the TKGs for entity alignment. Specifically, we design a richness-guided attention mechanism along with an adaptive weighting strategy to facilitate effective feature fusion. To ensure robust alignment despite noisy entity contexts, we introduce a dual-view neighborhood consensus algorithm that jointly refines the feature encoders to enforce local structural consistency of the predicted alignments. Extensive experiments demonstrate the superiority of \modelname{}, achieving state-of-the-art performance on public TEA benchmarks. 
\end{abstract}


\section{Introduction}
A Knowledge Graph (KG) is a structured knowledge base that captures real-world knowledge to support data-driven applications such as information retrieval \cite{kobayashi2000information}, information extraction \cite{sarawagi2008information, wang2017pdd}, domain adaption\cite{xiong2021source} and recommendation \cite{jiang2024challenging, lu2012recommender, zhao2025diversity}. With the growing use of KGs derived from diverse sources, integrating knowledge from multiple KGs has become crucial. Entity Alignment (EA), the process of identifying equivalent entities across KGs, plays a central role in KG fusion. 

Given the dynamic and complex nature of entity-wise interactions, the incorporation of temporal information into KGs has led to the emergence of Temporal Knowledge Graphs (TKGs) \cite{trisedya2019entity, zeng2020degree, ge2021make, xin2022ensemble, liu2023dependency}. TKGs extend traditional triples by including timestamps, enabling a richer representation of dynamic relationships among entities over time. However, this advancement also introduces the challenging task of Temporal Entity Alignment (TEA), as illustrated in Figure \ref{fig:alignment}. 
Existing TEA models \cite{xu2021time, xu2022time, liu2023unsupervised, cai2022simple, cai2023effective} often treat temporal features in a simplified way, lacking dedicated mechanisms for temporal refinement and noise mitigation. Specifically, \cite{xu2021time, xu2022time} treat temporal information as a type of relation for feature propagation and rely on straightforward feature concatenation during inference. In contrast, \cite{liu2023unsupervised, cai2022simple, cai2023effective} employ temporal encoders to facilitate temporal knowledge transfer, but they are often computationally expensive and sensitive to temporal heterogeneity. Although \cite{li2025htea} introduces a temporal attention mechanism and a heterogeneity refinement strategy to enhance temporal knowledge propagation, it does not explicitly address the challenges of complex feature interactions.

\begin{figure}
     \centering
     \includegraphics[width=\linewidth]
     {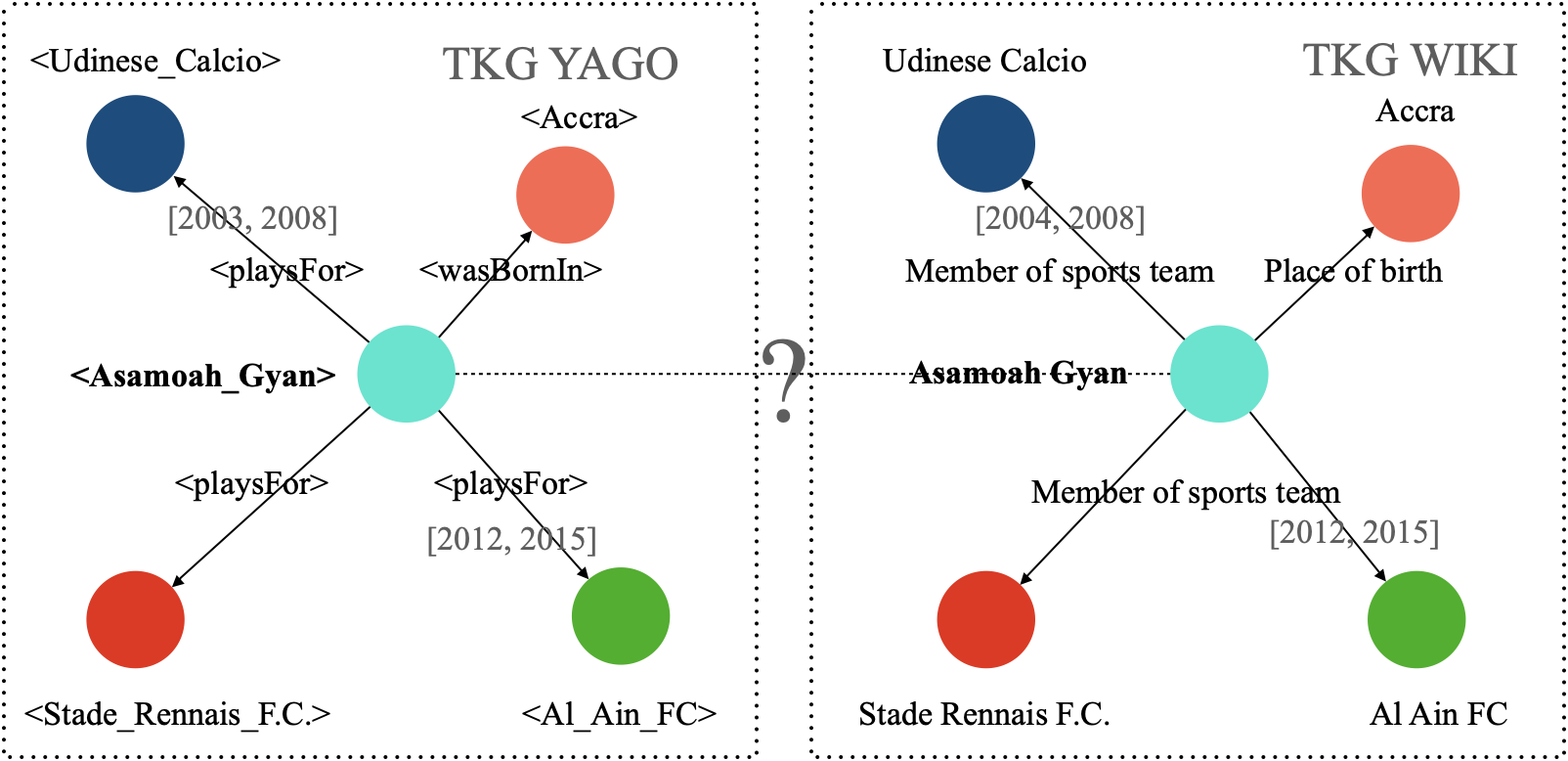}
     \caption{Example of temporal entity alignment.}
     \vspace{-8mm}
     \label{fig:alignment}
\end{figure}

\textbf{Issue 1: Feature richness-agnostic message passing.}
Existing TEA models are typically GNN-based, relying on message propagation between 
neighboring entities to refine representation learning. However, they often overlook the 
importance of feature richness in the local neighborhood when learning entity embeddings. 
To verify this, we conduct a preliminary study on the YAGO-WIKI180K dataset, grouping 
entities based on the richness levels of their core features — structural (number of 
neighboring entities or relation types) and temporal (number of temporal points/intervals) 
— and evaluate alignment accuracy (Hit@1) per group. As illustrated in 
Table~\ref{tab:richness}, entities with lower feature richness tend to generate less 
reliable alignments. Propagating their supervision signals to neighboring entities 
naturally hinders the model's ability to capture precise and informative features, 
resulting in suboptimal representations and inferior alignment results. This calls for 
a mechanism that distinguishes entities based on feature richness during the 
message-passing process.

\begin{table}[ht]
\caption{Impact of feature richness (Hit@1).}
 \vspace{-4mm}
\label{tab:richness}
\centering
\scriptsize
\begin{tabular}{c c c c}
\toprule
    \textbf{Feature} & \textbf{Low-richness} & \textbf{Medium-richness} & \textbf{High-richness} \\
    \midrule
    Entity &.0424 &.2109 &.2730 \\
    Relation &.0005 &.0009 &.0022 \\
    Temporal &.0026 &.0354 &.4675 \\
    \bottomrule
\end{tabular}
\end{table}

\begin{table}[ht]
\caption{Uniquely correct cases of feature encoders: Hop-1/2 sizes, temporal values indicate the number of hop-1/2 neighbors and different temporal values.}
\vspace{-4mm}
\label{tab:unique correct}
\centering
\scriptsize
\begin{tabular}{c c c c c}
\toprule
    \textbf{Feature} & \textbf{\# Correct} &\textbf{Hop-1 Size} & \textbf{Hop-2 Size} & \textbf{\# Temp} \\
    \midrule
    Structural &9525 &5 &758 &1 \\
    Temporal &1058 &5 &533 &3 \\
    Mixed &9464 &4.5 &533.5 &1 \\
    \bottomrule
\end{tabular}
\end{table}

\textbf{Issue 2: Static feature fusion without entity-specific weighting.}
Most TEA models combine structural and temporal features using basic methods such as 
concatenation or averaging, which limits their ability to adaptively emphasize the most 
informative signals for each entity pair. Structural and temporal features provide 
orthogonal but complementary insights for entity alignment, and their relative importance 
may vary depending on the specific entity context. To investigate this, we conduct a 
further preliminary study on the YAGO-WIKI180K dataset, comparing three feature encoders: 
structural, temporal, and mixed (simple concatenation of both). 
Table~\ref{tab:unique correct} reports the number of entity pairs that each encoder 
aligns correctly on its own while the others fail, where Hop-1/2 Size denotes the average 
number of entities in the hop-1/2 neighborhood reflecting structural richness, and \#Temp 
represents the average number of temporal values associated with the central entity 
indicating temporal richness. The results confirm a clear correlation between feature 
richness and model performance: the structural encoder is more effective for entities with 
rich neighborhoods, while the temporal encoder excels for temporally-diverse entities. 
This highlights the need for entity-specific, richness-aware weighting to integrate both 
feature modalities into the alignment decision.

To address these issues, we propose \modelname{}, a novel \underline{R}ichness-guided \underline{C}o-training framework for \underline{T}emporal \underline{E}ntity \underline{A}lignment. In particular, we design a novel richness-guided attention mechanism to facilitate effective message propagation in the feature encoders by quantifying the importance of each neighbor based on its structural and temporal richness. We also introduce a dynamic weighting strategy to adaptively integrate both structural and temporal features for embedding learning. To enhance model robustness, we propose a dual-view neighborhood consensus algorithm that jointly refines the feature encoders by enforcing local consistency in the predicted alignments. Our model can be further extended to a semi-supervised iterative training framework, where the training signals are progressively expanded with an effective bi-directional seed selection mechanism to eliminate noisy pseudo-labels. Our main contributions are summarized below:
\begin{itemize}[leftmargin=*]
\item We introduce novel dual-aspect feature encoders that fully leverage the orthogonal yet complementary features-structural and temporal-to learn more informative and comprehensive entity embeddings. The feature encoders are jointly refined via dual-view neighborhood consensus.
\item We highlight the importance of richness for feature encoders, and design richness-guided attention and adaptive weighting strategies to enhance the effectiveness of representation learning by dynamically prioritizing informative features.
\item Extensive experiments on publicly available TEA datasets clearly demonstrate the superiority of our \modelname{} model, achieving state-of-the-art alignment performance.
\end{itemize}

\section{Related Work}

Recent studies have recognized the importance of temporal information for entity alignment. Models like TEA-GNN \cite{xu2021time} and TREA \cite{xu2022time} treat temporal point embeddings as a medium for temporal representation, yet they tend to treat temporal information similarly to relational data, thus overlooking its unique properties. The following methods explore temporal information as an attribute: STEA \cite{cai2022simple} encodes temporal information using a temporal dictionary, while DualMatch \cite{liu2023unsupervised} incorporates a temporal encoder to enhance message propagation. MGTEA \cite{zeng2024benchmarking} further adopts a similar strategy to exploit fine‑grained temporal cues. However, these approaches are often computationally expensive, largely due to the need to construct and process temporal‑aspect similarity matrices. LightTEA \cite{cai2023effective}, an extension of LightEA \cite{mao2022lightea}, introduces a lightweight model for TKG alignment, but the temporal module leads to limited improvement on existing datasets \cite{xu2021time}. HTEA \cite{li2025htea} introduces an L1-weighted feature selection matrix and a temporal attention module to enhance the propagation of temporal information. Furthermore, a heterogeneity module is applied iteratively to alleviate the impact of heterogeneous temporal triples. Nevertheless, current TEA models still fail to effectively capture the interaction between structural and temporal features, and suffer from the noise issue on both sides. 
\begin{figure*}[htb]
    \centering
    \includegraphics[width=\linewidth]{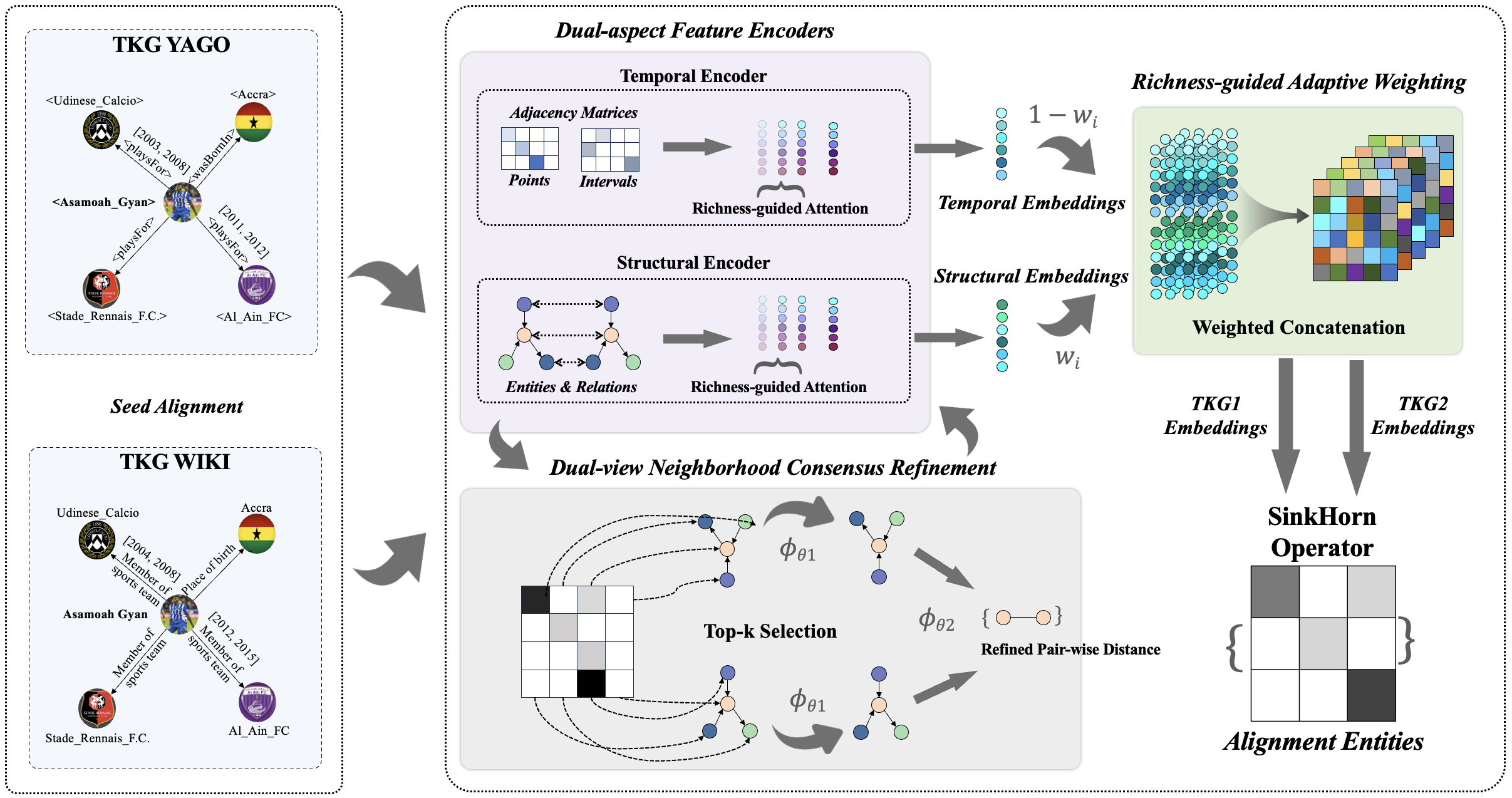}
    \vspace{-4mm}
    \caption{\modelname{} framework overview.}
    \vspace{-4mm}
    \label{fig:framework}
\end{figure*}

\section{Problem Formulation}
\begin{definition}[Temporal Knowledge Graph]
A temporal knowledge graph (TKG) is a directed graph, denoted by $\mathcal{G} = (\mathcal{E}, \mathcal{R}, \mathcal{T}, \mathcal{F})$, where $\mathcal{E}, \mathcal{R}$ and $\mathcal{T}$ represent the set of entities, relations and temporal points, respectively; $\mathcal{F} \subseteq \mathcal{E} \times \mathcal{R} \times \mathcal{E} \times \mathcal{T} \times \mathcal{T}$ is a collection of facts in the form of $f=(e_h, r, e_t, I)$. Specifically, $I=[t_s, t_e]$ is the temporal interval during which the head entity $e_h$ has a valid relation $r$ with the tail entity $e_t$, where $t_s \in\mathcal{T}$ (resp. $t_e \in\mathcal{T}$) represents the start time (resp. end time) of the fact. A placeholder ($\sim$) is used when $t_s$ or $t_e$ does not exist. 
$\mathcal{I} \subseteq \mathcal{T} \times \mathcal{T}$ denotes the set of temporal intervals.
\end{definition}
\vspace{-2mm}

\begin{definition}[Temporal Entity Alignment] Given two TKGs $\mathcal{G} = (\mathcal{E}, \mathcal{R}, \mathcal{T}, \mathcal{F})$ and $\mathcal{G}^\prime = (\mathcal{E}^\prime, \mathcal{R}^\prime, \mathcal{T}^\prime, \mathcal{F}^\prime)$ and a set of pre-aligned entity pairs $\mathcal{S} = \{(e, e^\prime)\ | e \equiv e^\prime, e \in \mathcal{E}, e^\prime \in \mathcal{E}^\prime\}$ as the seed alignment, where $\equiv$ denotes entity equivalence, \textit{temporal entity alignment} (TEA) aims to identify all possible newly matched entity pairs: $\{(e, e^\prime)\ |\ e \equiv e^\prime, e \in \mathcal{E}, e^\prime \in \mathcal{E}^\prime, (e, e^\prime)\notin \mathcal{S}\}$.
\end{definition}
\section{The \modelname{} Model}
Figure \ref{fig:framework} presents an overview of our proposed \modelname{} framework, which consists of three key modules: (1) Dual-aspect feature encoders that learns entity embeddings from both structural and temporal perspectives, along with a richness-guided attention mechanism to enhance the reliability of message propagation in the GNNs; (2) A richness-guided adaptive weighting strategy that balances both features with dynamically determined weights to generate a unified entity representation for alignment; (3) A denoising algorithm based on dual-view neighborhood consensus to refine the feature encoders by capturing locality-aware alignment patterns.

\subsection{Dual-aspect Feature Encoders}
As discussed earlier, the structural and temporal cues of a TKG provide orthogonal yet complementary insights for the alignment decision, making it necessary to build \textit{dual-aspect feature encoders} that learn both the separate feature representations and their joint effect. In particular, we decompose the structural and temporal perspectives of a TKG by introducing two neural models based on the Graph Attention Network (GAT) architecture for feature encoding. These include: a structural encoder $\mathcal{M}_{stru}$ that utilizes entity- and relation-level input to learn the structural and semantic representations of a TKG, and a temporal encoder $\mathcal{M}_{temp}$ that incorporates both temporal points and intervals to encode the temporal behaviors of entities. Each encoder is designed to receive and propagate distinct types of information from the TKG. Beyond separate feature learning, we also construct a mixed encoder $\mathcal{M}_{mix}$ that jointly encodes all types of features and balances them through an adaptively weighted combination.

\subsubsection{Structural and Temporal Encoders}
Given a TKG $\mathcal{G} = (\mathcal{E}, \mathcal{R}, \mathcal{T}, \mathcal{F})$, we extract four types of features (i.e., entities $\mathcal{E}$, relations $\mathcal{R}$, temporal points $\mathcal{T}$ and temporal intervals $\mathcal{I}$) and generate type-specific bipartite adjacency matrices $\mathbf{A}_{(\mathcal{C})}$ from the TKG, where $\mathcal{C}\in\{E, R, T, I\}$. $\mathbf{A}_{(\mathcal{C})}$ indicates the importance of each feature $c$ for learning the embedding of entity $e_i$. We utilize log-normalized initial representation to dampen the impact of high-frequency features \cite{li2025htea}:

\begin{equation}
a_{ic} = \frac{\left(\log(|\mathcal{F}_{ic}| + 1)\right)}{\sum_{c^\prime \in \mathcal{C}_{e_i}}\left(\log(|\mathcal{F}_{ic^\prime}| + 1)\right)}
\label{eq:bipartite_weight}
\end{equation}
$\mathcal{C}_{e_i}$ denotes the set of features contained by entity $e_i$; $\mathcal{F}_{ic}$ is the set of facts involving both entity $e_i$ and feature $c$. For example, when $\mathcal{C}=E$, $\mathcal{C}_{e_i}=\mathcal{N}_{e_i}$ represent all the neighboring entities that $e_i$ connects to. Similarly, when $\mathcal{C}=I$, $\mathcal{C}_{e_i}=\mathcal{I}_{e_i}$ indicates all the temporal intervals that $e_i$ exists a relation with other entities.

Both the structural and temporal encoders learn the type-specific entity embeddings layer-by-layer in a GAT-like manner. Specifically, the type-specific layer propagation is defined as follows:
\begin{align}
h_{e_{i(\mathcal{C})}}^{0} &= \left(\mathbf{A}_{(\mathcal{C})}[e_i]\right)^{\top} \mathbf{F}_{(\mathcal{C})}
\end{align}
\begin{align}
h_{e_{i}(\mathcal{C})}^{l+1} = \sigma(\sum_{e_j \in \mathcal{N}_{e_i}} \beta_{ij(\mathcal{C})}^l \mathbf{W}_{(\mathcal{C})}^lh_{e_{j}(\mathcal{C})}^l), l\geq0
\label{eq:feature_propogation}
\end{align}

Here, $\mathbf{F}{(\mathcal{C})}$ denotes the type-specific initialization matrix, while $\mathbf{A}{(\mathcal{C})}[e_i]$ represents the bipartite feature selection weights for entity $e_i$ with respect to type $\mathcal{C}$, as defined in Eq.~\ref{eq:bipartite_weight}. To enable message passing across neighboring entities, multi-layer feature propagation is performed as follows: $\beta_{ij}^l$ denotes the adaptive attention weight between entities $e_i$ and $e_j$ at layer $l$ (detailed in the next section), $\boldsymbol{W}_{(\mathcal{C})}^l$ is the learnable type-specific transformation matrix, and $\sigma(\cdot)$ is the ReLU activation function.

After $L$-hop feature propagation and aggregation, we concatenate the embeddings learned at each layer to obtain the final type-specific feature representation for entity $e_i$ as outlined below:
\begin{align}
    h_{e_i(\mathcal{C})} = [h_{e_i(\mathcal{C})}^{0}||h_{e_i(\mathcal{C})}^{1}||\cdots || h_{e_i(\mathcal{C})}^{L}]
\end{align}

The structural encoder $\mathcal{M}_{stru}$ concatenates the entity- and relation-specific embeddings to obtain the structural representation of entity $e_i$. Similarly, the temporal encoder $\mathcal{M}_{temp}$ combines the embeddings of both temporal points and temporal intervals, as below:
\begin{equation}
    h_{e_i(stru)} =[h_{e_i(E)}||h_{e_i(R)}]
\end{equation}
\begin{equation}
    h_{e_i(temp)} =[h_{e_i(T)}||h_{e_i(I)}]
\end{equation}

\subsubsection{Mixed Encoder}
Considering the complementary effect between structural and temporal perspectives of a TKG, we further build a mixed encoder $\mathcal{M}_{mix}$ to jointly obtain the dual-view feature representation of entity $e_i$ through a weighted concatenation of both features obtained by $\mathcal{M}_{stru}$ and $\mathcal{M}_{temp}$ respectively:
\begin{equation}
h_{e_i} = [w_ih_{e_i(stru)}\oplus (1-w_i)h_{e_i(temp)}]
\label{eq:weighted_combine}
\end{equation}

All encoders are trained with the same procedure. Given two TKGs $\mathcal{G} = (\mathcal{E}, \mathcal{R}, \mathcal{T}, \mathcal{F})$ and $\mathcal{G}^\prime = (\mathcal{E}^\prime, \mathcal{R}^\prime, \mathcal{T}^\prime, \mathcal{F}^\prime)$, along with their seed alignments $\mathcal{S} = \{(e_i, e_i^\prime)\ | e_i \equiv e_i^\prime, e_i \in \mathcal{E}, e_i^\prime \in \mathcal{E}^\prime\}$, we formulate a bi-directional probability-based loss over entity pairs, where both directions of alignment (i.e., from one TKG to the other and vice versa) are considered symmetrically to enhance robustness and alignment quality \cite{chen2023meaformer}.
\begin{equation}
p(e_i, e_i^\prime) = \frac{\gamma(e_i, e_i^\prime)}{\gamma(e_i, e_i^\prime) + \sum_{e_j \in \mathcal{N}^{ng}}\gamma(e_i, e_j)}
\end{equation}
where $\gamma(e_i, e_j) = exp(h_{e_i} \cdot h_{e_j}^T / \tau)$ and $\tau$ is the temperature hyperparameter. $\mathcal{N}^{ng}$ is the in-batch negative samples to enhance training efficiency. To consider both directions of entity alignment, the overall loss function is defined as follows:
\begin{equation}
\mathcal{L}_{align} = -log(p(e_i, e_i^\prime) + p(e_i^\prime, e_i))/2
\end{equation}

\subsection{Richness-guided Attention and Adaptive Weighting}
Existing neighborhood weighting methods usually rely on naive GAT mechanisms \cite{velickovic2018graph} or relation-based approaches \cite{mao2020relational}, which tend to be less interpretable and fail to capture distinctive feature contributions. Meanwhile, current TEA models \cite{xu2021time, xu2022time} typically incorporate various feature embeddings through a simple concatenation and feed the combined representations into alignment modules to predict entity mappings. However, our preliminary findings show that the richness of structural and temporal information varies significantly across entities, which should influence the weighting strategy during alignment. 

To capture this signal, we propose a \textit{richness-guided attention mechanism} along with an \textit{entity-wise adaptive weighting scheme} that balances the contributions of structural and temporal features. Recognizing that structural and temporal information contribute differently to each entity's representation, our approach adaptively prioritizes the more informative aspect for each entity, thereby enhancing alignment accuracy. Specifically, we introduce the concept of reference embedding and define a corresponding richness measurement, which quantifies the feature richness of each entity. These measurements form the foundation for both the richness-guided attention and adaptive weighting mechanisms used in our model.

\subsubsection{Reference Embeddings}
When counterpart entities exhibit similar structural or temporal characteristics (such as node degrees, neighborhood distributions or temporal behaviors), their embeddings can reliably capture underlying similarities. To leverage this, we propose an effective strategy for measuring structural and temporal richness. Specifically, we introduce a global reference entity $e_{ref}$ to simulate the entity with extremely low richness, e.g., without any neighbors or temporal behaviors. We obtain its corresponding type-specific representations, called reference embeddings $h_{e_{ref}(\mathcal{C})}$, similarly as the other entities using $\mathcal{M}_{stru}$ and $\mathcal{M}_{temp}$. These reference embeddings serve as representative proxy vectors that reflect the typical semantics of entities with low richness in the corresponding feature space. For example, the structural reference embeddings $h_{e_{ref}(\mathcal{C})}$ where $\mathcal{C} \in \{E,R\}$ quantify the degree-based richness of entity $e_i$, using the number of facts it participates in as the  count between $e_i$ and the reference entity $e_{ref}$ (Eq.\ref {eq:fact}). The reference entity's temporal aspect embeddings are none-temporal embeddings (point and interval), the correlation weight for entity, reference temporal embedding is calculated via (Eq.\ref{eq:bipartite_weight}) analogously.
\begin{align}
    |F_{ie_{ref}(E)}|, |F_{ie_{ref}(R)}| = \sum_{e_j \in \mathcal{N}_{e_i}}|F_{ie_j}|
\label{eq:fact}
\end{align}

\subsubsection{Richness Measurement}
We introduce a new metric, called \textit{reference similarity}, to quantify the feature richness of a specific entity $e_i$. Since the reference entity $e_{ref}$ is designed to simulate an entity prototype with extremely low richness, a low embedding similarity between $e_i$ and $e_{ref}$ naturally indicates that entity $e_i$ is structurally or temporally rich. To capture this signal, we estimate $e_i$'s feature richness by calculating its feature-specific reference similarity with regard to the corresponding reference embedding. Formally,
\begin{align}
\mathcal{S}_{e_i(\mathcal{C})} = \cos(h_{e_{i}(\mathcal{C})}, h_{e_{ref}(\mathcal{C})})
\label{eq:reference_distance}
\end{align}

\subsubsection{Richness-guided Feature Propagation}
During the multi-layer message propagation of our dual-aspect feature encoders (i.e., Eq. \ref{eq:feature_propogation}), we introduce a richness-guided attention weight $\beta_{ij(\mathcal{C})}^l$ to quantify the importance of each neighbor $e_j$ for learning the representation of entity $e_i$, as illustrate in Figure \ref{fig:attention weight}. The underlying intuition is that feature-specific richness varies across neighbors, while it is more valuable to propagate information from structurally or temporally rich neighbors. For example, our preliminary study (Table \ref{tab:richness}) proves that entities with higher degree are more likely to be correctly aligned, indicating that the training signals provided by these structurally rich neighbors are more reliable for learning robust entity embeddings. Similarly, a neighbor with diverse range of temporal features tends to carry more distinctive information for determining the entity alignment. Based on this observation, we formally define the richness-guided attention weight for each layer's feature propagation as follows:
\begin{align}
    \beta_{ij(\mathcal{C})}^l = \frac{\exp\left(-\mathcal{S}^l_{e_j(\mathcal{C})}\right)}{\sum_{e_k \in \mathcal{N}_{e_i}}\exp\left(-\mathcal{S}^l_{e_k(\mathcal{C})}\right)}, l\geq0
\end{align}
Here, $\mathcal{S}^l_{e_j(\mathcal{C})}$ at each feature propagation layer $l$ denotes the feature-specific reference distance obtained by Eq. \ref{eq:reference_distance}.

\begin{figure}
     \centering
     \includegraphics[width=\linewidth]
     {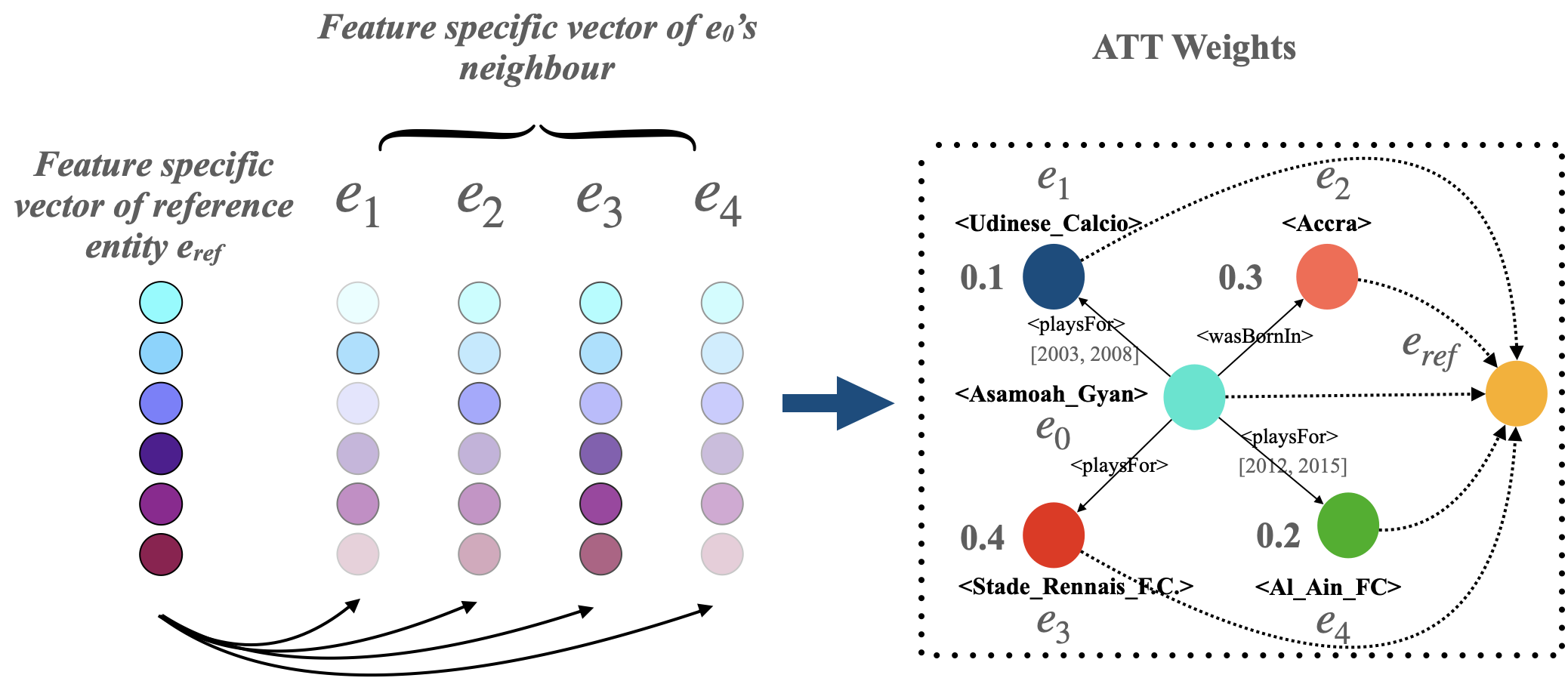}
     \caption{Example of richness-guided attention weights.}
     \label{fig:attention weight}
     \vspace{-4mm}
\end{figure}

\subsubsection{Richness-guided Feature Concatenation}
Our preliminary study (Table \ref{tab:unique correct}) also reveals that a straightforward concatenation of the structural and temporal features may degrade the alignment performance, calling for an adaptive weighting mechanism that can dynamically adjust the relative importance of both features for each specific alignment case. Similarly, these weights can be inferred based on the entity's structural and temporal richness. Inspired by \cite{chen2023meaformer} which applies weighted concatenation for balancing multi-modal features, we design a richness-guided adaptive weighting scheme to balance the contributions of structural and temporal features for learning the embedding of a specific entity.

Specifically, we utilize the reference similarity between an entity's type-specific embedding and the reference embedding to construct an initial \textit{reference matrix} $\mathbf{x}$: 
\begin{align}
\small
\mathbf{x} = 
\begin{bmatrix}
\mathcal{S}_{e_1(E)} & \mathcal{S}_{e_1(R)} & \mathcal{S}_{e_1(T)} & \mathcal{S}_{e_1(I)} \\
\mathcal{S}_{e_2(E)} & \mathcal{S}_{e_2(R)} & \mathcal{S}_{e_2(T)} & \mathcal{S}_{e_2(I)} \\
\vdots & \ddots & \vdots \\
\mathcal{S}_{e_n(E)} & \mathcal{S}_{e_n(R)} & \mathcal{S}_{e_n(T)} & \mathcal{S}_{e_n(I)} \\
\end{bmatrix}
\end{align}
The four columns represent the reference similarity for the four types of features, i.e., entities, relations, temporal points and intervals, respectively. 

We then apply a single multilayer perceptron (MLP) to perform a non-linear transformation on the initial reference matrix $\mathbf{x}$ to infer the adaptive entity-wise concatenation weights used in the mixed encoder $\mathcal{M}_{mix}$ (i.e., Eq. \ref{eq:weighted_combine}). Formally,
\begin{align}
\mathbf{h} = \phi(\mathbf{W}_1\mathbf{x} + \mathbf{b}_1), \quad
\mathbf{w} = \sigma(\mathbf{W}_2 \mathbf{h} + \mathbf{b}_2)
\end{align}
$\phi$ and $\sigma$ denotes MLP and sigmoid function, where $\mathbf{h}$ is the intermediate result, and $w_i \in \mathbf{w}$ is the weight obtained for each entity $e_i$. $\boldsymbol{W_1}$, $\boldsymbol{W_2}$, $\boldsymbol{b_1}$, $\boldsymbol{b_2}$ are learnable parameters.

\subsection{Neighborhood Consensus De-noising and Iterative Training}
To enhance the alignment compatibility among entities and their neighbors (i.e., neighbors of aligned entities should also be matched), we employ a neighborhood consensus model \cite{fey2020deep} to refine the dual-aspect feature encoders $\mathcal{M}_{stru}$ and $\mathcal{M}_{temp}$ for de-noising the learned entity embeddings, ensuring that aligned entities share similar local structures. Building on the complementary nature of structural and temporal features, we further introduce a \textit{dual-view neighborhood consensus strategy}, enabling mutual enhancement between the two modalities through joint training. 

In particular, given two TKGs $\mathcal{G} = (\mathcal{E}, \mathcal{R}, \mathcal{T}, \mathcal{F})$ and $\mathcal{G}^\prime = (\mathcal{E}^\prime, \mathcal{R}^\prime, \mathcal{T}^\prime, \mathcal{F}^\prime)$, we generate their structural and temporal feature embeddings using the corresponding encoders $\mathcal{M}_{stru}$ and $\mathcal{M}_{temp}$, and jointly refine both encoders by concatenating their learned embeddings:
\begin{align}
\mathbf{H} = [\mathbf{H}_{stru} \, || \, \mathbf{H}_{temp}]
\end{align}
We adopt the FAISS framework \cite{douze2024faiss} to perform efficient top-$k$ embedding retrieval from large-scale TKGs, as the top-$k$ candidates typically contain the most informative alignment signals:
\begin{equation}
     \mathbf{\hat{S}} = Softmax(Faiss_k [\mathbf{H} \cdot \mathbf{H}^\prime])
\end{equation}
where $\mathbf{\hat{S}}$ is the top-$k$ retrieved probability score matrix, and $exp$ is the softmax function. We then apply a multi-layer perceptron (MLP) to iteratively refine the similarity matrix, leveraging the injective node distances derived from the propagated features of the two TKGs to minimize the achieved embedding distances for two counterpart embeddings (Details in Appendix \ref{appendix:consensus}).

Our \modelname{} model can be further extended to an iterative training framework, denoted as \modelname{}$^+$, which progressively incorporates high-quality seeds for model training. To this end, we design a \textit{dual-aspect bi-directional seed selection} mechanism that exploits the complementary nature of the two orthogonal features—structural and temporal information—to collaboratively eliminate the influence of noisy pseudo-labels (Details in Appendix \ref{appendix:iterative}).
\section{Experiments}

\subsection{Experimental Settings}
\label{subsec:exp}
\noindent\textbf{Datasets.}
We conduct experiments on both homogeneous and heterogeneous TEA benchmarks (dataset statistics in Appendix \ref{appendix:dataset}). \cite{li2025htea} identifies a key limitation of existing TEA benchmarks—neglect of temporal heterogeneity across TKGs—and constructs a more realistic YAGO-WIKI180K dataset for evaluation. \cite{zeng2024benchmarking} introduces BTEA, a fine-grained TEA dataset that incorporates hybrid temporal triples and further refines temporal information to the date level. We present experimental results on YAGO-WIKI180K and BTEA in this section, while additional evaluations on the homogeneous TEA datasets are reported in Appendix \ref{appendix:homogeneous}.

\noindent\textbf{Baselines.}
We compare \modelname{} with eight state-of-the-art TEA approaches, including two structural embedding models and six temporal-aware models. 1) \textit{time-unaware}: Dual-AMN \cite{mao2021boosting} and LightEA \cite{mao2022lightea}; 2) \textit{time-aware}: TEA-GNN \cite{xu2021time}, STEA \cite{cai2022simple}, DualMatch \cite{liu2023unsupervised}, LightTEA \cite{cai2023effective}, MGTEA\cite{zeng2024benchmarking}, and HTEA\cite{li2025htea}. We ignore models that utilize various side information (entity names, entity descriptions, attributes) for a fair comparison.

\noindent\textbf{Parameter Settings.}
We report the detailed parameter settings for the YAGO-WIKI180K and BETA datasets to support better reproducibility of the experimental results.

For the YAGO-WIKI180K dataset, we set the embedding dimension to $d = 50$, GNN depth to $l = 2$, number of attention layers to $L = 1$, batch size to $b = 512$, and dropout rate to $\text{drop} = 0.3$. The model is optimized using RMSprop with a learning rate $\eta = 0.005$. For the neighborhood consensus module, we select the top-$k = 15$ most similar entities. The randomized embedding dimension is set to $d_r = 32$, and the number of MLP propagation steps to $s = 5$, with the MLP comprising 2 layers and a hidden dimension of 32. For the dynamic weighting module, the MLP also consists of 2 layers but uses a smaller hidden dimension of 4. In the semi-supervised setting, the number of training iterations is set to $t = 2$ to balance effectiveness and efficiency. The training schedule includes $E_{\text{rel}} = 20$, $E_{\text{temp}} = 60$, $E_{\text{mix}} = 5$, and $E_{\text{neigh}} = 90$ epochs for the respective modules.

For the BETA dataset, the main differences lie in the neighborhood consensus component, where we set the MLP propagation steps to $s = 25$ and select the top-$k = 45$ most similar entities. The training epochs are adjusted accordingly: $E_{\text{rel}} = 8$, $E_{\text{temp}} = 16$, $E_{\text{mix}} = 1$, and $E_{\text{neigh}} = 110$.

\noindent\textbf{Evaluation Metrics.}
Building on previous research, we use Mean Reciprocal Rank (MRR) and Hits at Top-N (H@N) as evaluation metrics. H@N measures the precision of the top-N retrieved entities, while MRR is calculated using the average reciprocal rank of all entities.

All experiments are conducted on a Linux cluster equipped with an AMD EPYC 3 Milan CPU, 
an NVIDIA H100 GPU, and 1,500GB of RAM. Each experiment is repeated five times and the 
mean results are reported. Code is available at our public repository.\footnote{\url{https://github.com/psych-mind/RCTEA_}}

\begin{table*}[ht]
\caption{Overall TEA performance on YAGO-WIKI180K (2000 seeds). The p-value is the result of one sample t-test between \modelname{} and their corresponding strong baselines.}
\label{tab:mainresult}
\centering
\small
\begin{tabular}{c ccc ccc ccc ccc}
\toprule
\multirow{2}{*}{\textbf{Model}} 
& \multicolumn{3}{c}{\textbf{All}} 
& \multicolumn{3}{c}{\textbf{Non-tem $(50.4\%)$}} 
& \multicolumn{3}{c}{\textbf{Sparse-tem $(43\%)$}} 
& \multicolumn{3}{c}{\textbf{Dense-tem $(6.6\%)$}} \\
\cmidrule(lr){2-4} \cmidrule(lr){5-7} \cmidrule(lr){8-10} \cmidrule(lr){11-13}
& MRR & H@1 & H@10 & MRR & H@1 & H@10 & MRR & H@1 & H@10 & MRR & H@1 & H@10 \\
\midrule
Dual-AMN    & .381 & .330 & .476 & .159 & .110 & .253 & .572 & .521 & .668 & .824 & .768 & .928\\
LightEA     & .376 & .329 & .464 & .169 & .121 & .260 & .562 & .516 & .647 & .742 & .691 & .838 \\
\midrule
STEA        & .308 & .277 & .366 & .007 & .004 & .112 & .496 & .457 & .570 & .911 & .880 & .963\\
TEA-GNN     & .352 & .307 & .436 & .145 & .101 & .227 & .519 & .472 & .605 & .843 & .794 & .931\\
LightTEA    & .361 & .317 & .442 & .150 & .108 & .231 & .548 & .505 & .628 & .744 & .696 & .833\\
DualMatch   & .388 & .342 & .474 & .157 & .103 & .259 & .577 & .538 & .650 & .912 & .882 & .966\\
MGTEA &.405 &.355 &.499 &.185 &.126 &.300 &.583 &.541 &.661 &.920 &.893 &.967\\ 
HTEA\footnotemark     & .404 & .354 & .500 & .183 & .128 & .290 & .591 & .544 & .679 & .876 & .846 & .930\\
\textbf{\modelname{}} & \textbf{.459} & \textbf{.413} & \textbf{.545} & \textbf{.226} & \textbf{.168} & \textbf{.336} & \textbf{.655} & \textbf{.618} & \textbf{.723} & \textbf{.960} & \textbf{.947} & \textbf{.983}\\
\textit{p}-value &1e-4 &3e-5 &7e-4 &5e-3 &2e-3 &2e-3 &7e-5 &3e-5 &3e-4 &3e-3 &2e-3 &1e-2 \\
\midrule
\textbf{\modelname{}$^+$ (iter)} & \textbf{.468} & \textbf{.423} & \textbf{.554} & \textbf{.232} & \textbf{.173} & \textbf{.344} & \textbf{.669} & \textbf{.633} & \textbf{.735} & \textbf{.971} & \textbf{.960} & \textbf{.987}\\
\textit{p}-value &6e-5 &2e-5 &4e-4 &3e-3 &1e-3 &2e-3 &3e-5 &1e-5 &1e-4 &2e-3 &1e-3 &1e-2\\
\bottomrule
\end{tabular}
\vspace{-4mm}
\end{table*}
\footnotetext{The reported results for HTEA differ from those in the original paper, as we standardize the output embedding size across all models to ensure a fair comparison.}

\begin{table}[ht]
\caption{Overall TEA performance on BETA .}
\label{tab:mainresult}
\centering
\small
\begin{tabular}{c cc cc}
\toprule
\multirow{2}{*}{\textbf{Model}} 
& \multicolumn{2}{c}{BETA (\textbf{$10\%$})} 
& \multicolumn{2}{c}{BETA (\textbf{$30\%$})} \\
\cmidrule(lr){2-3} \cmidrule(lr){4-5}
  & H@1 & H@10 & H@1 & H@10 \\
\midrule
Dual-AMN  &.566 &.695  &.589 &.757    \\
LightEA   &.492 &.673  &.595 &.744  \\
\midrule
STEA   &.506 &.637  &.556 &.681 \\
TEA-GNN  &.557 &.709  &.609 &.751  \\
LightTEA   &.615 &.727  &.666 &.769\\
DualMatch   &.650 &.765  &.689 &.796\\
MGTEA  &\textbf{.686} &.795 &.711 &.813 \\
HTEA   &.647 &.783 &.660 &.772   \\

\textbf{\modelname{}}   &.681 &\textbf{.797} &\textbf{.713} &\textbf{.816} \\
\textit{p}-value  &3e-2 &1e-2 &5e-3 &5e-3\\
\midrule
\textbf{\modelname{}$^+$ (iter)} &\textbf{.707} &\textbf{.813} &\textbf{.720} &\textbf{.819}  \\
\textit{p}-value &8e-3 &7e-3 &3e-3 &4e-3\\
\bottomrule
\end{tabular}
\vspace{-4mm}
\end{table}

\subsection{Main Results}
We consider the following research question in this section: \textbf{RQ1: Can \modelname{} enhance TEA performance over existing baselines?} We compare \modelname{} with existing TEA models and report their alignment performance on the overall YAGO-WIKI180K dataset and various subsets of different characteristics: \textit{Dense-tem} with rich temporal features, \textit{Sparse-tem} with scarce temporal information, and \textit{Non-tem} without any temporal facts. 
As presented in Table \ref{tab:mainresult}, \modelname{} consistently achieves the state-of-the-art performance on all variants of the YAGO-WIKI180K dataset, surpassing HTEA and MGTEA, the second-best models, by nearly $6\%$ on all evaluation metrics. In particular, \modelname{} outperforms HTEA and MGTEA on the Non-tem, Sparse-tem and Dense-tem dataset variants by around $4\%$, $7\%$ and $5\%$, respectively, attributed to our design of richness-guided attention and adaptive weighting mechanisms that fully utilize the feature richness to enhance representation learning. The alignment accuracy further increases when the \modelname{} model is iteratively trained with high-quality pseudo-labels (\# iterations = 2 by default), i.e., \modelname{}$^+$ in Table \ref{tab:mainresult}.

For the BETA dataset, even without leveraging fine-grained temporal information (e.g., date annotations), our model achieves superior performance across all seed settings at both 
$10\%$ and $30\%$. In particular, \modelname{} attains competitive Hit@1 performance in low-resource settings and outperforms other baselines in the remaining settings, despite not employing the time-consuming temporal encoders used in MGTEA and DualMatch. This further demonstrates the effectiveness and efficiency of our approach. Moreover, the iterative variant \modelname{}$^+$ also yields notable performance gains, especially in the low-resource setting.

\subsection{Ablation Study}
\label{exp:ablation}
We further explore the contribution of each component in our \modelname{} framework, and present the detailed results in Table \ref{tab:ablation}.

\begin{table}[ht]
\caption{Ablation study on YAGO-WIK180K (H@1).}
\label{tab:ablation}
\centering
\small
\begin{tabular}{c c c c c}
\toprule
    \textbf{Model} & \textbf{All} & \textbf{Non} & \textbf{Sparse} & \textbf{Dense}\\
    \midrule
    \textbf{\modelname{}} &.4135 &.1684 &.6183 &.9465\\
    \modelname{} w/o. E &.1933 &.0013 &.3300 &.6758\\
    \modelname{} w/o. R &.3588 &.1033 &.5713 &.9213\\
    \modelname{} w/o.T &.3874 &.1377 &.6004 &.9031\\
    \modelname{} w/o. I &.3807 &.1349 &.5905 &.8868 \\
    \midrule
    \modelname{} w/o. ET &.3937 &.1416 &.6082 &.9147\\
    \modelname{} w/o. RT &.3936 &.1381 &.6117 &.9197\\
    \modelname{} w/o. TT &.3942 &.1416 &.6095 &.9170\\
    \modelname{} w/o. IT &.3922 &.1381 &.6094 &.9129\\
    \modelname{} w. GAT &.3825 &.1402 &.5881 &.8887 \\
    \modelname{} w. RST &.3919 &.1520 &.5986 &.8726 \\
    \midrule
    \modelname{} w. EW &.3919 &.1371 &.6099 &.9142\\
    \modelname{} w/o. DW &.3951 &.1427 &.6103 &.9166\\
    \midrule
    \modelname{} w/o. NC &.4065 &.1607 &.6120 &.9414\\
    \modelname{} w/o. DUAL &.4120 &.1710 &.6125 &.9416\\
    \bottomrule
\end{tabular}
 \vspace{-4mm}
\end{table}

\noindent\textbf{RQ2: Is it necessary to combine all features in \modelname{}?} We evaluate the importance of each type of features, i.e., entities, relations, temporal points, and temporal intervals, by excluding them separately from the \modelname{} model. As observed from the first part of Table \ref{tab:ablation}, the alignment performance drops in each case, highlighting the effectiveness of our dual-aspect feature encoders to integrate both structural and temporal perspectives for representation learning.

\noindent\textbf{RQ3: Can the richness-guided attention mechanism enhance alignment performance?} We examine the effectiveness of richness-guided attention weights for feature propagation by selectively disabling them. Variants \modelname{} w/o. ET, RT, TT, and IT represent the removal of attention weights for entities, relations, temporal points, and temporal intervals, respectively. The experimental results confirm that each feature-specific attention contributes positively to the overall model, validating the effectiveness of our multi-view attention design. We further compare the proposed richness-guided attention with other attention mechanisms used by existing TEA models, including naïve GAT attention (\modelname{} w. GAT) and relation-specific attention (\modelname{} w. RST). Our approach consistently outperforms both variants, showcasing the importance of feature richness for attention learning.

\noindent\textbf{RQ4: Is adaptive weighting necessary for feature fusion?} We evaluate the effectiveness of the dynamic weighting strategy by modifying our mixed feature encoder $\mathcal{M}_{mix}$ with equal weights (\modelname{} w. EW) or removing the dynamic weighting module entirely (\modelname{} w/o. DW). Experimental results demonstrate that our adaptive weighting mechanism enhances overall alignment performance, surpassing the equally weighted version by a large extent.

\noindent\textbf{RQ5: What is the impact of dual-view neighborhood consensus algorithm?} Finally, we investigate the importance of neighborhood consensus for entity alignment. \modelname{} w/o. NC disables the neighborhood consensus module, while \modelname{} w/o. DUAL removes the dual-view feature representation used in the consensus process. Experimental results indicate that both components are crucial for effective feature refinement and alignment inference.

\subsection{Efficiency Study}
Finally, we compare \modelname{} with existing TEA approaches in terms of model efficiency, considering the following two research questions: 
\textbf{RQ6: Can \modelname{} balance alignment accuracy and model efficiency?}
\textbf{RQ7: Can the iterative training of \modelname{}$^+$ further enhance alignment performance, and at what cost?} From Figure \ref{fig:efficiency}(a), we can see that \modelname{} achieves comparable time consumption with the most efficient model LightTEA and HTEA, while further enhances the alignment accuracy, which validates our model's superiority. Figure \ref{fig:efficiency}(b) indicates that the training time increses linearly with the iterative training, and the accuray improves mostly between the first and second iterations. 

\begin{figure}[htbp]
    \centering
    \includegraphics[width=\linewidth]{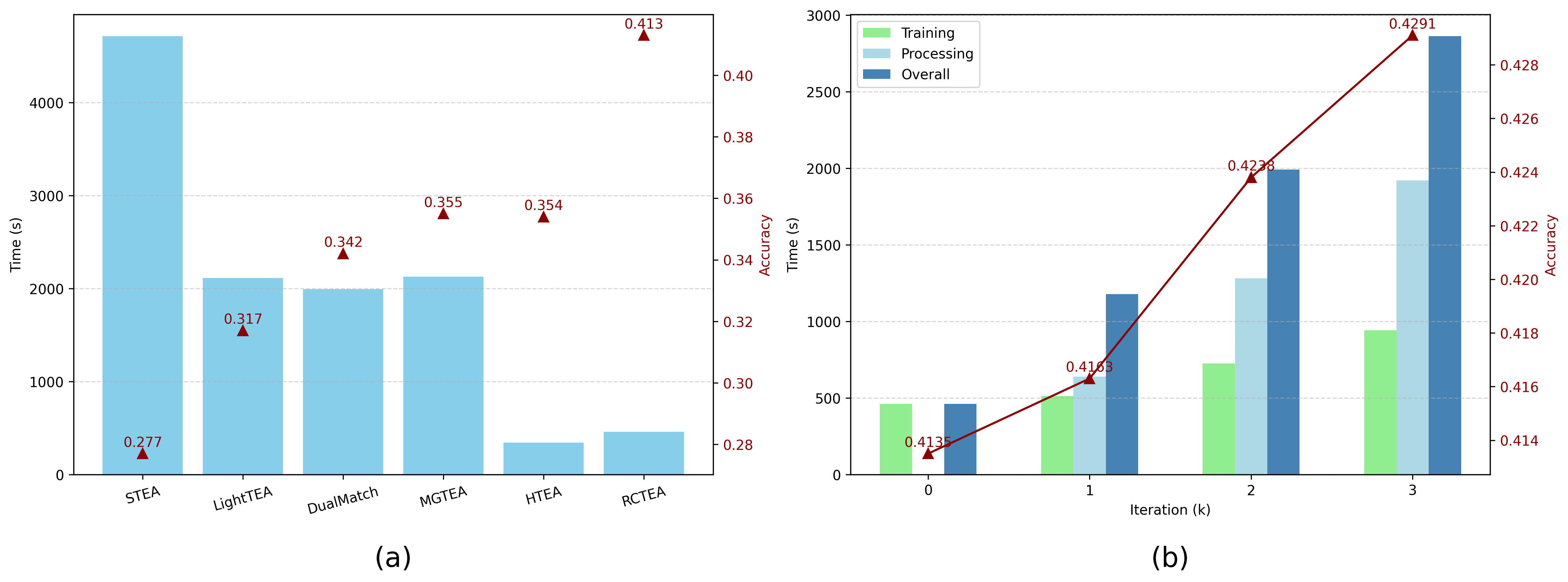}
    \vspace{-4mm}
    \caption{Efficiency evaluation on YAGO-WIKI180K. Specifically: (a) compares the time consumption and overall performance across different models; (b) illustrates the relationship between model iterations and performance for our approach.}
    \label{fig:efficiency}
    \vspace{-4mm}
\end{figure}

\section{Conclusion}
Existing TEA models tend to underestimate the importance of relational and temporal features, as well as the inherent noise present in both. Our proposed \modelname{} model features an efficient split-combined training framework and a dynamic weighting mechanism that adaptively balances different feature aspects. To mitigate noise, we introduce a joint neighborhood consensus strategy that enhances the representativeness of both views. Additionally, our dual-aspect seed selection method significantly reduces noise in general seed selection. Extensive experiments on TEA datasets validate the efficiency and effectiveness of our approach. 

\section{Acknowledgements}
This work was partially supported by the Australia Research Council under the Discovery Project (Grant No. DP230101753).
\section*{Limitation}
This paper explores a richness measurement to guide graph neural networks in propagating and integrating informative features, thereby providing explicit alignment signals to the model. However, when the TKGs evolve over time, the corresponding feature richness also needs to be incrementally refined, which may incur significant extra cost for frequently-updated TKGs. Hence, extending the proposed \modelname{} model to evolving TKGs constitutes an important direction for future work. Moreover, investigating how to effectively balance or integrate various attention mechanisms—such as standard GAT, relation-specific attention, and other variants—is another promising research direction.


\bibliography{custom}

\appendix

\newtheorem{theorem}{Theorem}

\section{Model Details}

\subsection{Neighborhood Consensus De-noising}
\label{appendix:consensus}
To enhance alignment consistency between entities and their local neighborhoods, we incorporate a neighborhood consensus model \cite{fey2020deep} to refine the dual-aspect encoders $\mathcal{M}_{stru}$ and $\mathcal{M}_{temp}$. This approach mitigates noise in the learned embeddings by encouraging aligned entities to exhibit similar neighborhood structures. Leveraging the complementary nature of structural and temporal features, we further propose a \textit{dual-view neighborhood consensus strategy} that enables mutual reinforcement between the two modalities through joint training. This design facilitates bidirectional enhancement between structural and temporal signals, thereby improving the quality and robustness of the final entity representations.

In particular, given two TKGs $\mathcal{G} = (\mathcal{E}, \mathcal{R}, \mathcal{T}, \mathcal{F})$ and $\mathcal{G}^\prime = (\mathcal{E}^\prime, \mathcal{R}^\prime, \mathcal{T}^\prime, \mathcal{F}^\prime)$, we first generate their structural and temporal feature embeddings using the corresponding encoders $\mathcal{M}_{stru}$ and $\mathcal{M}_{temp}$:
\begin{align}
\mathbf{H}_{stru}, \mathbf{H}_{stru}^\prime &= \mathcal{M}_{stru}(\mathbf{A}_{(E)}, \mathbf{A}_{(R)}) \\
\mathbf{H}_{temp}, \mathbf{H}_{temp}^\prime &= \mathcal{M}_{temp}(\mathbf{A}_{(T)}, \mathbf{A}_{(I)})
\end{align}
$\mathbf{A}_{(\mathcal{C})}$ is the bipartite adjacency matrix of the TKG for a specific feature type $\mathcal{C}\in\{E,R,T,I\}$.

Instead of relying on top-$k$ entity retrieval from a single encoder to de-noise the corresponding feature based on neighborhood consensus, we jointly refine both encoders by concatenating their learned embeddings:
\begin{align}
\mathbf{H},\mathbf{H}^\prime = [\mathbf{H}_{stru} \, || \, \mathbf{H}_{temp}],[\mathbf{H}_{stru}^\prime \, || \, \mathbf{H}_{temp}^\prime]
\end{align}
This approach leverages the complementary perspectives of the two encoders to enhance overall alignment performance. We adopt the FAISS framework \cite{douze2024faiss} to perform efficient top-$k$ embedding retrieval from large-scale TKGs, as the top-$k$ candidates typically encapsulate the most informative signals for entity alignment:
\begin{equation}
     \mathbf{\hat{S}} = Softmax(Faiss_k [\mathbf{H} \cdot \mathbf{H}^\prime])
\end{equation}
The resulting top-$k$ probability score matrix, denoted by $\mathbf{\hat{S}}$, is computed using a softmax function over the retrieved similarities.

Finally, the consensus model utilizes a multi-layer perceptron (MLP) to iteratively refine the similarity matrix, leveraging the injective node distances derived from the propagated features of the two TKGs to minimize the achieved embedding distances for two counterpart embeddings, as detailed below:
\begin{align}
\boldsymbol{O}, \boldsymbol{O}^\prime  = \boldsymbol{\Phi}_{\theta_2}(\boldsymbol{I}_{|\mathcal{E}|}, \boldsymbol{A}_{(E)}),\ \boldsymbol{\Phi}_{\theta_2}(\boldsymbol{\hat{S}}^{T}\boldsymbol{I}_{|\mathcal{E}|}, \boldsymbol{A}_{(E)})
\end{align}
\begin{align}
  \hat{S}_{i, j}^{l + 1} = \hat{S}_{i, j}^{l} + \mathbf{\Phi}_{\theta_2}(d_{i, j}), 
  \ \mathrm{where} \ l\geq0
\end{align}
\begin{align}
\mathbf{S^{L}} = Softmax(\mathbf{\hat{S}^{L}})
\end{align}
Here, $d_{ij} = \vec{o_i} - \vec{o_j}$ measures the distance between two nodes, $\mathbf{A}_{(E)}$ is the adjacency matrix, $\mathbf{O}$ and $\mathbf{O}^\prime$ are the reconstructed embeddings for the two TKGs, $d_{ij}$ is the distance between the two reconstructed counterpart embeddings, $\mathbf{I} \sim \mathcal{N}(0, 1)$ is used for randomized embeddings, $\mathbf{\Phi}_{\theta_1}$ is a simple GNN, and $\mathbf{\Phi}_{\theta_2}$ is the defined MLP.

To train the dual-view neighborhood consensus module, we maximize the log-normalized probability scores of the seed alignments for refining the structural and temporal encoders, as defined below:
\begin{equation}
\mathcal{L}_{nc}  = -\sum_{(e_i,e_i^\prime)\in\mathcal{S}}log(S^{L}_{i, i^\prime})
\end{equation}

\subsection{\modelname{}$^+$ with Iterative Training}
\label{appendix:iterative}
Our \modelname{} model can be further extended to an iterative training framework, denoted as \modelname{}$^+$, which progressively incorporates high-quality seeds for model training. To this end, we design a \textit{dual-aspect bi-directional seed selection} mechanism to exploit the complementary nature of the two orthogonal features—structural and temporal information—to collaboratively eliminate the influence of noisy pseudo-labels. In particular, we employ the two specialized feature encoders: $\mathcal{M}_{stru}$ for identifying structural-based seed alignments and $\mathcal{M}_{temp}$ for temporal-based matched pairs. To maintain high-quality seed selection, we impose bi-directional constraints \cite{mao2020mraea} derived from both structural and temporal views, ensuring the consistency and reliability of the selected seeds. The complete seed selection procedure is outlined in Algorithm 1:

\begin{algorithm}[ht]
\caption{Seed Selection}
\label{algo:simple_seed_selection}
\KwIn{
    $S_r, S_t$: structural and temporal similarity matrices,
    $P$: pre-aligned seed pairs
}
\KwOut{
    $\hat{P}$: selected seed pairs
}

$\hat{P} \leftarrow P$;

\For{$e \in E$}{

$e_r^\ast \gets \arg\max_{e'} S_r(e,e'), \quad
 e_t^\ast \gets \arg\max_{e'} S_t(e,e')$\;
    
    \If{
        $\mathrm{argmax}_{e''} S_r[e'', e_r^*] = e$ \textbf{and}
        $\mathrm{argmax}_{e''} S_t[e'', e_t^*] = e$ \textbf{and}
        $e_r^* = e_t^*$
    }{
        $\hat{P} \leftarrow \hat{P} \cup \{(e, e_r^*)\}$\;
    }
}
\Return $\hat{P}$\;
\end{algorithm}

\subsection{Sinkhorn Operator}
The Sinkhorn operator utilizes a fast and effective way to compute the assignment problem. It iteratively normalizes rows and columns to generate results in a doubly stochastic matrix, as follows:
\begin{align*}
    &Sinkhorn^{0}(S) = \exp(S),\\
    &Sinkhorn^{m}(S) = \mathcal{N}_c(\mathcal{N}_r(Sinkhorn^{m-1}(S))),\\
    &Sinkhorn(S) = \lim_{x\to\infty}\ Sinkhorn^{m}(S)
\end{align*}
where $\mathcal{N}_r(S) = S \oslash (S1_{N}1_{N}^{T})$, $\mathcal{N}_c(S) = S \oslash (1_{N}1_{N}^{T}S)$ are row and column-wise normalization operators of a matrix, $\oslash$ represents the element-wise division, and $1_{N}$ is a column vector of ones.

\section{Proof of the Richness Measurement}
 
\begin{theorem}[Effect of Neighborhood Diversity on Reference Weight Contribution]
Let an entity have a neighborhood degree of fixed size $N$ and neighborhood number $n$, where $w_r$ of its count are reference entities. We define the log-normalized weight assigned to the reference embedding as:
\begin{align*}
\alpha_r = \frac{\log(w_r + 1)}{\sum_{i=1}^n \log(w_i + 1)}
\end{align*}
Consider two cases:
\begin{itemize}
    \item \textbf{Low-diversity case:} All $N - w_r$ non-reference degrees are identical (i.e., low diversity).
    \item \textbf{High-diversity case:} All $N - w_r$ non-reference degrees are unique (i.e., high diversity).
\end{itemize}
Then, under log-normalization, the weight $\alpha_r$ assigned to the reference embedding is strictly smaller in the high-diversity case than in the low-diversity case, provided that $N - w_r > 1$. Consequently, the contribution of the reference embedding to the target entity representation is lower under high-diversity conditions, increasing the distance between the target entity representation and the reference embedding.
\end{theorem}

\begin{proof}
Let $N$ be the total number of neighbors' degrees (fixed), and $w_r$ be the count of reference degrees.
We define the normalized weight assigned to the reference embedding in each case.

\noindent\textbf{Case 1: Low Diversity} (all $N - w_r$ non-reference degrees are the same):
\[
\alpha_r^{\text{low}} = \frac{\log(w_r + 1)}{\log(w_r + 1) + \log(N - w_r + 1)}
\]
\textbf{Case 2: High Diversity} (all $N - w_r$ non-reference degrees are distinct):
\[
\alpha_r^{\text{high}} = \frac{\log(w_r + 1)}{\log(w_r + 1) + (N - w_r)\log(2)}
\]
We want to show that:
\[
\alpha_r^{\text{high}} < \alpha_r^{\text{low}}
\]
This is equivalent to proving:
\[
(N - w_r)\log 2 > \log(N - w_r + 1)
\]
This inequality holds for $N - w_r > 1$, since for all $x \geq 2$:
\[
x \log 2 > \log(x + 1)
\]
Therefore:
\[
\alpha_r^{\text{high}} < \alpha_r^{\text{low}} \quad \text{whenever} \quad N - w_r > 1
\]
Thus, increasing neighborhood diversity (with a fixed reference count $w_r$) lowers the normalized weight $\alpha_r$, reducing the contribution of the reference embedding and increasing its relative distance to the target entity representation.
\end{proof}

\begin{theorem}[Increasing Neighborhood Size Reduces Reference Embedding Contribution]
Let the final embedding $e$ of an entity be computed from a set of $n$ neighborhood embeddings $\{e_i\}_{i=1}^n$ and a reference embedding $e_r$ as follows:
\begin{align*}
e = \frac{1}{Z} \left( \sum_{i=1}^{n} \log(w_i + 1) e_i + \log(w_r + 1) e_r \right)
\quad \\
\text{where} \quad Z = \log(w_r + 1) + \sum_{i=1}^{n} \log(w_i + 1)
\end{align*}
Assume each neighborhood embedding $e_i$ is a mixture:
\[
e_i = \alpha e_r + (1 - \alpha) e_i^{\text{other}}, \quad \text{where } \alpha \in [0, 1)
\]
Then, the total contribution of the reference embedding $e_r$ to $e$ is:
\[
W_{e_r} = \frac{\log(w_r + 1) + \sum_{i=1}^{n} \alpha \log(w_i + 1)}{\log(w_r + 1) + \sum_{i=1}^{n} \log(w_i + 1)}
\]
If we assume all neighborhood counts are equal, i.e., $\log(w_i + 1) = x > 0$ for all $i$, and $\log(w_r + 1) = A$, then:
\[
W_{e_r}(n) = \frac{A + n \alpha x}{A + nx}
\]
Then, $W_{e_r}(n)$ is a strictly decreasing function of $n$ for fixed $A > 0$, $x > 0$, and $\alpha \in [0, 1)$. That is, as the number of neighbors increases, the contribution of the reference embedding decreases.
\end{theorem}

\begin{proof}
Let us define:
\[
W(n) = \frac{A + n \alpha x}{A + n x}
\]
We compute the derivative of $W(n)$ with respect to $n$:
\[
\begin{aligned}
\frac{dW}{dn} 
&= \frac{d}{dn} \left( \frac{A + n \alpha x}{A + n x} \right)\\
&= \frac{\alpha x (A + n x) - x (A + n \alpha x)}{(A + nx)^2}
\end{aligned}
\]
We simplify the numerator:
\begin{align*}
&\alpha x (A + nx) - x (A + n \alpha x) \\
&= \alpha x A + \alpha n x^2 - x A - n \alpha x^2 \\
&= \alpha x A - x A \\
&= x A (\alpha - 1) < 0 \quad \text{since } \alpha < 1
\end{align*}
Thus, $\frac{dW}{dn} < 0$, meaning $W(n)$ is strictly decreasing with $n$.
Therefore, as the number of neighbors increases, the normalized contribution of the reference embedding $e_r$ to the final embedding $e$ decreases.
\end{proof}

\begin{theorem}[Smaller Reference Portion Reduces Cosine Similarity with Reference Embedding]
Let $e_r \in \mathbb{R}^d$ be a unit-norm reference embedding, i.e., $\|e_r\| = 1$.  
Let $e^{\text{other}} \in \mathbb{R}^d$ be an embedding orthogonal to $e_r$, i.e., $\langle e_r, e^{\text{other}} \rangle = 0$.  
We define the mixed embedding as:
\[
e(\alpha) = \alpha e_r + (1 - \alpha) e^{\text{other}}, \quad \text{with } \alpha \in [0, 1]
\]
Then, the cosine similarity between $e(\alpha)$ and $e_r$ is strictly increasing with $\alpha$.  
In other words, decreasing the reference portion $\alpha$ results in a lower cosine similarity with the reference embedding.
\end{theorem}

\begin{proof}
We compute the cosine similarity between $e(\alpha)$ and $e_r$:
\[
\cos(\theta) = \frac{\langle e(\alpha), e_r \rangle}{\|e(\alpha)\| \cdot \|e_r\|} 
= \frac{\langle \alpha e_r + (1 - \alpha)e^{\text{other}}, e_r \rangle}{\|e(\alpha)\|}
\]
Using orthogonality: $\langle e_r, e^{\text{other}} \rangle = 0$, and $\|e_r\| = 1$, we get:
\[
\cos(\theta) = \frac{\alpha}{\|e(\alpha)\|}
\]
Compute the norm:
\[
\|e(\alpha)\|^2 = \|\alpha e_r + (1 - \alpha)e^{\text{other}}\|^2 
= \alpha^2 + (1 - \alpha)^2 \|e^{\text{other}}\|^2
\]
Let $C = \|e^{\text{other}}\|^2 > 0$. Then:
\[
\cos(\theta) = \frac{\alpha}{\sqrt{\alpha^2 + (1 - \alpha)^2 C}}
\]
Now consider the function:
\[
f(\alpha) = \frac{\alpha}{\sqrt{\alpha^2 + (1 - \alpha)^2 C}}
\]
This function is strictly increasing with $\alpha \in [0, 1]$ for any constant $C > 0$. Therefore, decreasing $\alpha$ reduces the cosine similarity between $e(\alpha)$ and $e_r$.
\end{proof}

\section{Experimental Settings}

\begin{table*}[t]
\caption{Dataset Statistics: \#Entities, \# Relations, and \#Facts indicate the number of entities, relations, and facts in two TKGs. Additionally, $|\mathcal{S}|$: number of ground-truth aligned entity pairs; $|\mathcal{P}|$: number of temporal points; $\lambda_t$: proportion of temporal triples; $\lambda_{TH}$: proportion of heterogeneous temporal triples.} 
\label{tab:dataset}
\small
\begin{tabular}{ccccccccc}
    \toprule
    Dataset & \#Entities & \#Relations & \#Facts & $|\mathcal{S}|$ &$|\mathcal{P}|$ & $\lambda_{t_1}$ & $\lambda_{t_2}$ & $\lambda_{TH}$ \\
    \midrule
    DICEWS &9,517 | 9,573 & 247 | 246 & 307,552 | 307,553 & 8,566  & 4,017 & 1 & 1 & 0 \\
    YAGO-WIKI50K & 49,629 | 49,222 & 11 | 30 & 221,050 | 317,814 & 49,172 & 245 &1 &1 &0 \\
    YAGO-WIKI20K & 19,493 | 19,929 & 32 | 130 & 83,583 | 142,568 & 19,462 & 405 &.634 &.825 &0 \\
    YAGO-WIKI180K & 187,987 | 187,977 & 32 | 261 & 924,935 | 1,636,020 & 187,977 & 1,000 &.218 &.265 &.082\\
    BETA &42,666 | 42,297 &257 | 45 &199,879 | 162,320 &40,364 &967 &.645&.350&.390\footnotemark \\
    \bottomrule
\end{tabular}
\end{table*}
\footnotetext{BETA's heterogeneous portion of temporal triples is measured at the date level.}

\subsection{TEA Datasets}
\label{appendix:dataset}
In this section, we present the specifications and construction procedures of various publicly available TEA datasets, including the homogeneous datasets (DICEWS, YAGO-WIKI50K, YAGO-WIKI20K) and the more realistic and challenging heterogeneous datasets (YAGO-WIKI180K and BETA). Table \ref{tab:dataset} compares the statistics of the five TEA datasets.

\textbf{DICEWS} \cite{xu2021time}. Integrated Crisis Early Warning System (ICEWS) is a publicly available, large-scale event-based database that contains political events with specific time annotations extracted from millions of real-world news stories. ICEWS05-15 is a subset of ICEWS, consisting of 100,904 entities, 251 relations, 4,017 timestamps, and 461,329 facts from 2005 to 2015. It is commonly used as a TKG benchmark dataset in the community. DICEWS is built based on ICEWS05-15. Initially, ICEWS05-15 is randomly divided into two subsets, $\mathcal{Q}_1$ and $\mathcal{Q}_2$, of similar size, with an overlap ratio of 50\% in the number of shared quadruples between $\mathcal{Q}_1$ and $\mathcal{Q}_2$. However, this dataset includes all temporal triples and inherently ignores the temporal heterogeneity issue \cite{li2025htea}, as all non-temporal triples are excluded from the common ICEWS05-15 dataset.

\textbf{YAGO-WIKI50K} \cite{xu2021time}. Wikidata is a free and open knowledge base that stores structured data from Wikipedia. Similarly, YAGO is an open-source knowledge base extracted from Wikipedia and WordNet. Both contain a large number of identical entities represented in different surface forms, and some facts are associated with temporal information in various formats, such as timestamps and time intervals. \cite{lacroix2018canonical} constructed a large-scale TKG dataset from Wikidata, consisting of  432,715 entities, 407 relations, and 1,724 timestamps (retaining only year information) by filtering out high-frequency entities and relations. The entire dataset includes over 7 million triples in total, with around 10\% associated with temporal information. 
Based on this, YAGO-WIKI50K is built from YAGO and Wikidata, starting with the top 50,000 entities selected based on their frequencies in the dataset and linked to their corresponding Wikidata counterparts using QIDs. Two TKGs are generated by filtering out facts where entities appear in only one KG. Temporal information is then attached from the Wikidata part to the YAGO counterpart triples, and non-temporal triples are removed from the filtered TKGs. The YAGO-WIKI50K dataset is a fully temporal dataset without temporal heterogeneity, as the temporal information in the YAGO part is derived from the Wikidata counterpart triples.

\textbf{YAGO-WIKI20K} \cite{xu2022time}. YAGO-WIKI20K is constructed using a similar procedure as YAGO-WIKI50K, with the primary difference being that the number of selected Wikidata entities is reduced to 20,000 while retaining the non-temporal facts in the two TKGs. Consequently, YAGO-WIKI20K is a temporally-hybrid dataset yet without the issue of temporal heterogeneity.

\begin{table*}
\caption{TEA results on DICEWS and WY50K. The best results are highlighted in bold. Underline indicates the second-best results.}
\label{tab:homo result}
\centering
\small
\begin{tabular}{c ccc ccc ccc ccc}
\toprule
\multicolumn{1}{c}{} & \multicolumn{3}{c}{\textbf{DICEWS (1K)}} & \multicolumn{3}{c}{\textbf{DICEWS (200)}} & \multicolumn{3}{c}{\textbf{YAGO-WIKI50K (5K)}} & \multicolumn{3}{c}{\textbf{YAGO-WIKI50K (1K)}}\\
\cmidrule(rl){2-4} \cmidrule(rl){5-7} \cmidrule(rl){8-10} \cmidrule(rl){11-13}
\textbf{MODEL} & {MRR} & {H@1} & {H@10} & {MRR} & {H@1} & {H@10} & {MRR} & {H@1} & {H@10} & {MRR} & {H@1} & {H@10}\\
\midrule
Dual-AMN &.779 &.716 &.893 &.733 &.668 &.854 &.922 &.897 &.964 &.834 &.755 &.890\\
LightEA &.833 &.785 &.918 &.779 &.721 &.878 &.960 &.948 &.979 &.902 &.878 &.945\\
\midrule
STEA &.941 &.928 &.960 &.941 &.927 &.961 &.954 &.935 &.986 &.916 &.887 &.966\\
TEA-GNN &.911 &.887 &.947 &.902 &.876 &.941 &.909 &.879 &.961 &.775 &.723 &.871\\
TREA &.933 &.914 &.966 &.927 &910 &960 &.958 &.940 &.989 &.885 &.840 &.937\\
LightTEA &.959 &.952 &.970 &.955 &.949 &.966 &\textbf{.990} &\textbf{.986} &\textbf{.997} &\underline{.977} &\underline{.969} &\underline{.989}\\
DualMatch &\textbf{.961} &\underline{.953} &\textbf{.973} &\textbf{.961} &\textbf{.953} &\textbf{.974} &.986 &.981 &\underline{.996} &.961 &.947 &.984 \\
MGTEA & - & - & - & \underline{.960} & \underline{.951} & \textbf{.974} & - & - & - & .960 & .947 & .982 \\
HTEA &.943 &.932 &.959 &.928 &.914 &.949 &.979 &.971 &.991 &.950 &.931 &.980\\
\textbf{\modelname} &.955 &.946 &.970 &.947 &.936 &.964 &\underline{.987}
&\underline{.982} &.995 &.968 &.957 &.987 \\
\midrule
\textbf{\modelname$^+$(iter)} &\textbf{.961} &\textbf{.954} &\textbf{.973} &.958 &.950 &\underline{.970} &\underline{.987} &.981 &\underline{.996} &\textbf{.979} &\textbf{.971} &\textbf{.992} \\
\bottomrule
\vspace{-6mm}
\end{tabular}
\end{table*}

\textbf{YAGO-WIKI180K} \cite{li2025htea}. The YAGO-WIKI180K dataset is constructed using the same Wikidata source as the YAGO-WIKI50K and YAGO-WIKI20K datasets, with a key distinction in how temporal information is handled. The YAGO triples are extracted from YAGO3 \cite{suchanek2007yago}. All available YAGO triples are first retrieved, followed by the extraction of temporal information from associated file records that contain timestamp details for each fact. 
This procedure enables the temporal information to be derived directly from the YAGO portion, rather than incorporating temporal data from the WIKI counterpart. Using available QID links, the YAGO and WIKI TKGs are aligned by selecting triples in which YAGO entities have corresponding counterparts in the WIKI graph. The resulting dataset includes approximately 19,000 shared entities, with 32 and 261 relations on the YAGO and WIKI sides, respectively, and includes a partial set of temporal triples. Compared to other existing TEA datasets, this construction process yields nearly four times as many entities. Furthermore, due to the independent sourcing of temporal information from the two KGs, the dataset exhibits temporal heterogeneity in 8.2\% of the aligned triples. This level of heterogeneity more accurately reflects the variation commonly encountered in real-world settings.

\textbf{BETA} \cite{zeng2024benchmarking}. Unlike all the other datasets which standardizes timestamps to year-level granularity, BETA retains multi-granular temporal information, with $28\%$ of temporal facts in the Wikidata subset and $14\%$ in the YAGO subset containing precise month-day annotations. This fine-grained temporal data introduces more diverse information in alignment tasks. In terms of entity distribution, BETA exhibits a more realistic spread: while most entities in YAGO-WIKI50K are associated with a single dominant temporal relation (e.g., member\_of\_sports\_team), BETA contains entities with no temporal facts as well as those linked to multiple temporal relations similar to YAGO-WIKI180K. This distribution supports the design of new alignment scenarios beyond the original single-relation setting. Additionally, relation frequency in BETA is more balanced: whereas member\_of\_sports\_team and plays\_for account for over $98\%$ of all quadruples in YAGO-WIKI50K, the top five relations in BETA collectively contribute to less than $35\%$, reflecting a broader and more domain-general knowledge structure. These characteristics make BETA a more challenging and representative benchmark for the TEA task.

\section{Experimental Results}

\subsection{Results on Homogeneous TEA Datasets}
\label{appendix:homogeneous}
In this section, we evaluate the performance of our model, \modelname, on two widely used homogeneous TEA datasets: DICEWS and YAGO-WIKI50K, with certain number of seed alignments. As shown in Table~\ref{tab:homo result}, \modelname\ achieves competitive performance in these relatively easy datasets, ranking among the top two models across various experimental settings.  Notably, both the iterative and base versions of \modelname\ exhibit strong representational capacity, even without relying on an explicit time-consuing temporal encoder, which significantly enhances model efficiency~\cite{mao2022lightea, liu2023unsupervised}.

\subsection{Richness-stratified Analysis}

In this section, we present an experiment to evaluate the effectiveness of our proposed richness measurement. Specifically, we compare the similarity between type-specific entity embeddings and their corresponding type-level reference embeddings. For each feature type, entities are divided into three groups—low, medium, and high richness—based on their respective richness scores. In particular, we define three richness levels: low, medium, and high, according to the following bins:
\begin{itemize}
    \item \textbf{Entity (E):} $[0,3)$, $[3,8)$, $[8,+\infty)$
    \item \textbf{Relation (R):} $[0,2)$, $[2,4)$, $[4,+\infty)$
    \item \textbf{Temporal point (T):} $[0,1)$, $[1,5)$, $[5,+\infty)$
    \item \textbf{Temporal interval (I)} $[0,3)$, $[3,5)$, $[5,+\infty)$
\end{itemize}

\begin{table}
\caption{Grouped similarity of feature-specific embedding and reference embedding in YAGO-WIKI180K (Metric: Average cosine similarity).}
 \vspace{-2mm}
\label{tab:grouped cos sim}
\centering
\footnotesize
\begin{tabular}{c c c c}
\toprule
    \textbf{Feature} & \textbf{High-rich} & \textbf{Med-rich} & \textbf{Low-rich} \\
    \midrule
    Entity &-0.0240 &-0.0137 &0.0737 \\
    Relation &-0.2080 &-0.0424 &-0.0376 \\
    Temporal point &-0.0452 &0.0298 &0.4373 \\
    Temporal interval &0.0202 &0.0358 &0.2862 \\
    \bottomrule
\end{tabular}
 \vspace{-2mm}
\end{table}

\begin{table}
\caption{Edge weight after richness-guided attention.}
\vspace{-2mm}
\label{tab:richness distance}
\tiny
\centering
\footnotesize
\begin{tabular}{c c c}
\toprule
    \textbf{Entity} & \textbf{Neighbour Size} & \textbf{Attention} \\
    \midrule
    \scriptsize{Mexico\_national\_football\_team} &307 &0.3457 \\
    \scriptsize{Toros\_Neza} & 50 &0.3348 \\
    \scriptsize{Mexico\_national\_under-17} & 33 & 0.3195 \\
    \bottomrule
\end{tabular}
 \vspace{-2mm}
\end{table}

Table~\ref{tab:grouped cos sim} reveals that entities with higher feature richness tend to exhibit smaller similarity compared to the reference embedding, whereas feature-sparse entities remain larger. This observation supports our hypothesis that feature-rich entities possess more diverse and distinctive embedding represent than the typical one. Furthermore, the proposed similarity metric proves to be effective in quantifying the potential information richness of feature-specific embeddings.

\subsection{Case Study for Richness-guided Attention}
In this section, we focus on the richness-guided attention mechanism, specifically in the context of the entity feature. As a case study, we examine the entity \textit{Éver\_Guzmán} from the YAGO TKG, which corresponds to index 81890.

Details of this entity's associated triples are illustrated in Figure~\ref{fig:case study richness}, where it is shown to have three neighbors. As presented in Table~\ref{tab:richness distance}, our attention mechanism successfully captures the nuanced differences among the neighboring entities, with the assigned attention weights being approximately proportional to the number of neighbors associated with each.

\subsection{Case Study for Adaptive Weighting}
In this section, we analyze the dynamic weighting behavior of structural and temporal features during training in YAGO-WIKI180K. We focus on a representative example—entity index 2000—along with its ground-truth counterpart and the average Relation feature weights across all entities from iteration 30 to 39. As shown in Figure~\ref{fig:dynamic-weighting study} (a), relation features consistently receive slightly higher weights than temporal features, reflecting their richer and less noisy characteristics. Although the weights fluctuate across iterations, they remain centered around 0.5, indicating the model’s ability to adaptively balance feature contributions. Figure~\ref{fig:dynamic-weighting study} (b) shows the weight difference between the selected entity and its counterpart, with similar trends suggesting the model captures consistent patterns that aid accurate alignment. This study highlights the effectiveness of the dynamic weighting strategy throughout the training process.

\begin{figure}
    \centering
    \includegraphics[width=\linewidth]{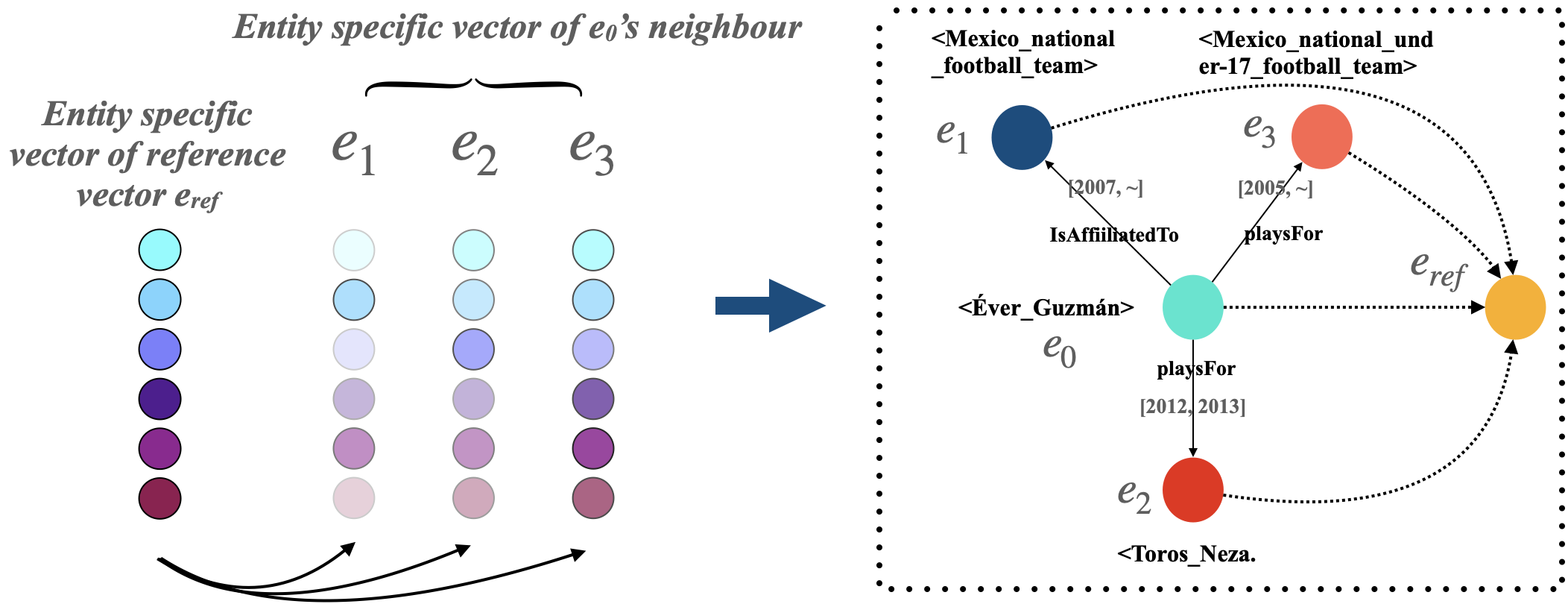}
    \caption{The triple details for entity \textit{Éver\_Guzmán}.}
    \label{fig:case study richness}
\end{figure}

\begin{figure}
    \centering
    \includegraphics[width=\linewidth]{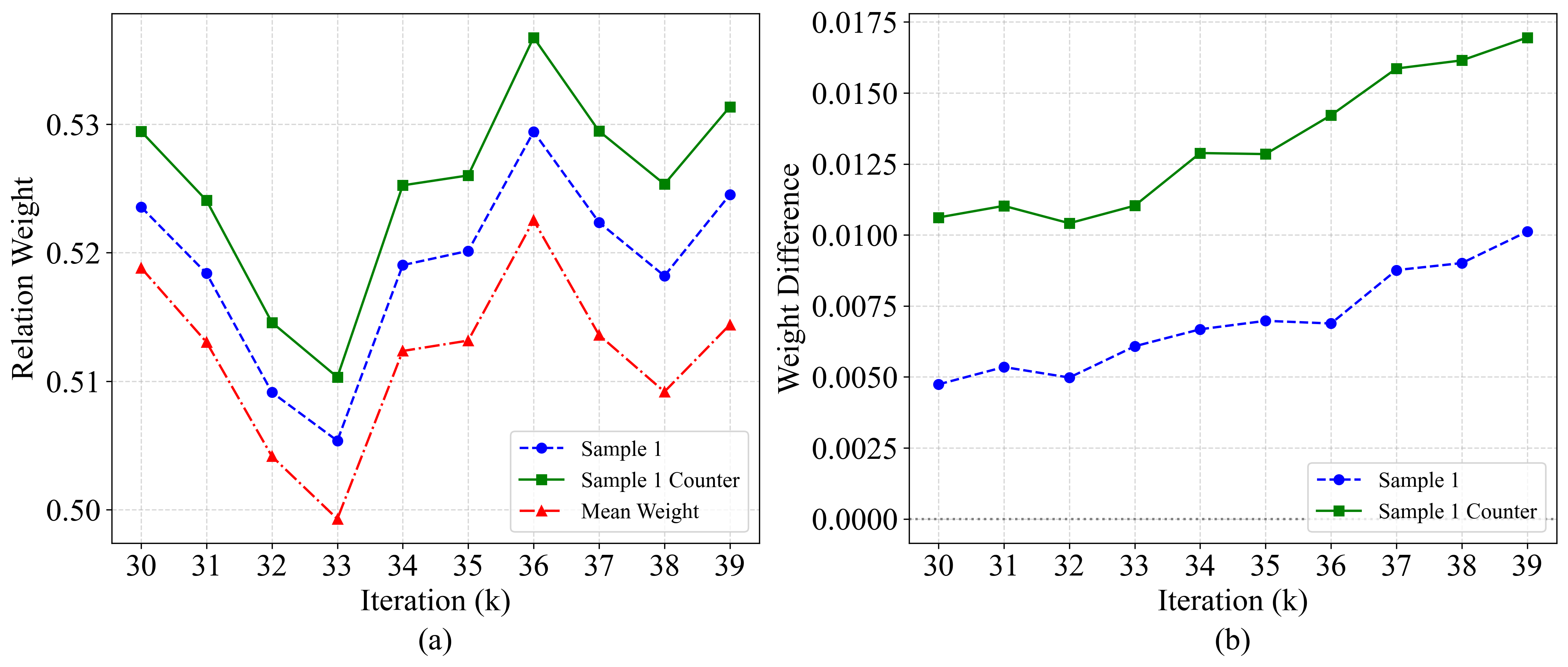}
    \caption{(a) Relation weighting for selected entity vs. Iteration. (b) Average weighting difference vs. Iteration.}
    \label{fig:dynamic-weighting study}
   \vspace{-4mm}
\end{figure}

\begin{figure}
   \centering
   \includegraphics[width=\linewidth]{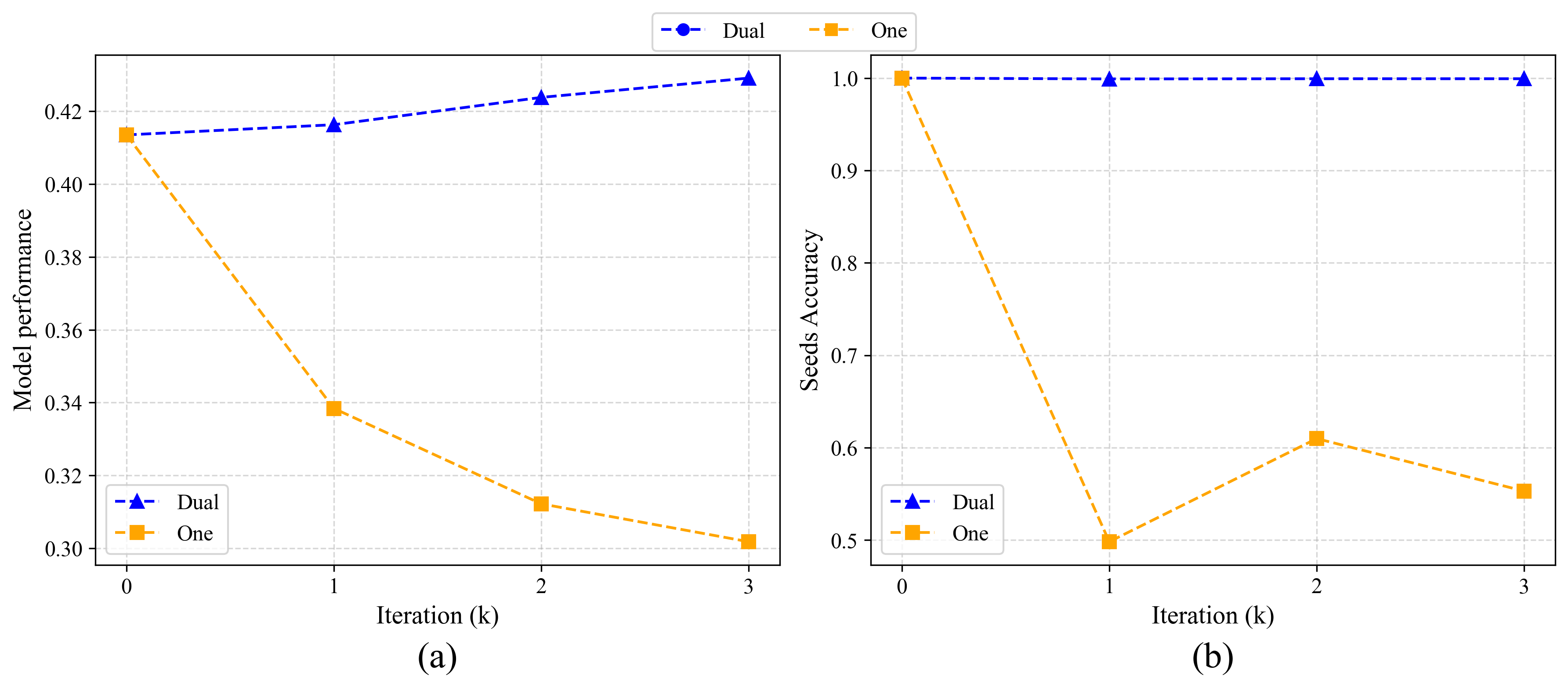}
   \caption{(a) Model Performance vs. Iteration, illustrating the impact of iterations on overall model performance. (b) Seed Accuracy vs. Iteration, showing how iterations affect overall seed accuracy.}
   \label{fig:seed}
   \vspace{-4mm}
\end{figure}

\subsection{Seed Accuracy Analysis}
We investigate the effectiveness of dual-aspect bi-directional seed accuracy in this section. Specifically, we compare the performance of using $\mathcal{M}_{stru}$ and $\mathcal{M}_{temp}$ for seed refinement (Dual) versus only using $\mathcal{M}_{mix}$ for seed refinement in YAGO-WIKI180K. As shown in Figure \ref{fig:seed}(a), the overall performance consistently improves with increasing iterations under dual-aspect refinement. In contrast, using only the $\mathcal{M}_{mix}$ model for seed refinement leads to a significant decline in performance. This trend is also evident in Figure \ref{fig:seed}(b), where the seed accuracy remains nearly perfect for the dual-view split-model approach but deteriorates substantially when relying solely on the $\mathcal{M}_{mix}$ model. These results highlight the critical role of seed quality in determining overall model performance.



\end{document}